\documentclass[12pt,leqno]{amsart}

\usepackage{amsmath,amsthm,amscd,amssymb,enumitem}
\usepackage[foot]{amsaddr}
\usepackage{xcolor}
 \usepackage{graphicx} 
\usepackage[PostScript=dvips,balance,midshaft,nohug]{diagrams}

\newtheorem{thm}[subsection]{Theorem}
\newtheorem*{thm*}{Theorem}
\newtheorem{cor}[subsection]{Corollary}
\newtheorem{lem}[subsection]{Lemma}
\newtheorem*{lem*}{Lemma}
\newtheorem*{mono*}{Monotonicity Lemma}
\newtheorem{prop}[subsection]{Proposition}
\newtheorem*{prop*}{Proposition}

\newtheorem{conj*}{Conjecture}
\theoremstyle{definition}
\newtheorem{defn}[subsection]{Definition}
\newtheorem{remark}[subsection]{Remark}

\renewcommand\mod[1]{\ (\mathop{\rm mod}#1)}
\newcommand{\quash}[1]{}

\newcommand{\Diag}{\text{Diag}}

\newcommand{\CC}{\mathbb C}\newcommand{\PP}{\mathbb P}
\newcommand{\RR}{\mathbb R}\newcommand{\QQ}{\mathbb Q}

\newcommand{\ZZ}{\mathbb Z}
\newcommand{\JJ}{\mathbb J}

\renewcommand{\Diag}{\mathop{\rm diag}}

\newcommand{\TT}{\mathbb T}

\newcommand{\Hess}{\mathop{\rm{Hess}}}
\newcommand{\Cu}{{\mathbb J}}  %%% probability current

\newcommand{\sslash}{\mathbin{/\mkern-6mu/}}

\newcounter{remarkscounter}

\newcommand{\ct}[1]{\textcolor{black}{#1}}

\begin{document}

\title{Nodal count for a random signing of a graph with disjoint cycles}
\author{Lior Alon}\address{Massachusetts Institute of Technology, Cambridge MA}
%\author[]{Lior Alon, Ram Band, Gregory Berkolaiko and}
\author{Mark Goresky}\address{Institute for Advanced Study, Princeton N.J.}

\begin{abstract}
Let $G$ be a simple, connected graph on $n$ vertices, and further assume that $G$ has disjoint cycles (see \S 3). Let $h$ be a real symmetric matrix supported on $G$ (for example, a discrete Schr\"odinger operator). The eigenvalues of $h$ are ordered increasingly, $\lambda_1 \le \cdots \le \lambda_n$, and if $\phi$ is the eigenvector corresponding to $\lambda_k$, the nodal (edge) count $\nu(h,k)$ is
the number of edges $(rs)$ such that $ h_{rs}\phi_{r}\phi_{s}>0$.  The nodal surplus is $\sigma(h,k)
= \nu(h,k) - (k-1)$. Let $h'$ be a random signing of $h$, that is a real symmetric matrix obtained from $h$ by changing
the sign of some of its off-diagonal elements. If $h$ satisfies a certain generic condition (cf.~\S 1.2) we show for each $k$ that the nodal surplus has a binomial distribution $\sigma(h',k)\sim Bin(\beta,\frac{1}{2})$.
%The signings of $h$ are the real symmetric matrices obtained from $h$ by changing
%the sign of some of its off-diagonal elements, and their nodal surplus is defined similarly. %We show, for each eigenvalue $\lambda_k$ that the nodal
%surplus of a random signing is binomially distributed.

%We show that the distribution of values $\sigma(h',k)$ is binomial,
%as $h'$ varies over the signings of $h$.  
Part of the proof
follows ideas developed by the first author together with Ram Band and
Gregory Berkolaiko in a joint unpublished project studying a similar
question on quantum graphs. 
\end{abstract}
\maketitle

\section{Introduction}\label{sec-intro}
\subsection{}
Let $G=G([n],E)$ be a simple graph on $n$ ordered vertices $[n]:=\{1,2,\ldots,n\}$ with a set of edges $E$. Write $r \sim s$ if the vertices $r \ne s$ are connected by an edge
$(rs) \in E$.  An $n \times n$ matrix
$h$ is {\em supported} (resp.~strictly supported) on $G$ if for any $r\ne s$,  $h_{rs} \ne 0 \implies r \sim s$ (resp. $h_{rs} 
\ne 0\iff r \sim s$ for $r\ne s$). 
Let $\mathcal S(G)$ (resp.~$\mathcal H(G)$) denote the vector space of real symmetric 
(resp.~Hermitian) matrices supported on $G$.  The eigenvalues of such a symmetric matrix  $h\in \mathcal S(G)$ 
are real and ordered,  $\lambda_1(h) \le \lambda_2(h) \le \cdots \le \lambda_n(h)$.
We say that $\phi\in \RR^n$ is \emph{nowhere-vanishing} if $\phi_{j}\ne0$ for all $j$. If $\phi$ is a nowhere-vanishing eigenvector of
$h$, with simple eigenvalue $\lambda_k$, then its {\em nodal (edge) count} is
\[ \nu(h,k)=\left|\{(rs)\in E\ :\ \phi_r h_{rs} \phi_s >0\}\right|.\]
(If $h_{rs}<0$, as in the case of the graph Laplacian or more generally, a discrete Schr\"odinger operator,  the nodal (edge) count is the number of edges on which $\phi$ changes sign.)  
If the graph $G$ is a tree, the nodal count is exactly $\nu(h,k)=k-1$ \cite{Fiedler}, however, this is not the case if $G$ is not a tree \cite{Band}. Consequently the {\em nodal surplus} for the $k$-th eigenvalue of a $h$ is defined to be 
\[\sigma(h,k):= \nu(h,k) - (k-1),\]
and it was proven to be non-negative and bounded by $\beta = |E|-n+1 ={\mathrm {rank}}(H_1(G))$, the first Betti number of $G$ 
\[0 \le \sigma(h,k) \le \beta.\]

A {\em signing} of $h\in \mathcal S(G)$ is a symmetric matrix $h'$ obtained from $h$ by changing the sign of some of its off-diagonal elements. \ct{When considering a random signing $h'$, we choose an element from the set of $2^{|E|}$ signings uniformly at random. In this way, $\sigma(h',k)$ is a random variable supported on $\{0,1,\ldots,\beta\}$. In this paper, for generic $h$ supported on a graph $G$ with disjoint cycles, and for each $k\in[n]$, we determine the distribution of $\sigma(h',k)$.} \quash{A fundamental problem in graph theory is to determine the distribution of nodal counts, or
nodal surplus for the eigenvalue $\lambda_k(h')$, as $h'$ varies over all signings of $h$. }

For a further introduction to these ideas, we refer the reader to \cite{AG}.

Following numerical simulations and the quantum graph analog \cite{Universality}, it was conjectured in \cite{AG} that generically, as $\beta\to \infty$ the distribution of the nodal surplus is expected to obey a universal law, converging to a Gaussian centered at $\beta/2$ with variance of order $\beta$. It was shown to hold for complete graphs with matrices that have a dominant diagonal. 
\subsection{}\label{subsec-star}
In this paper, we consider a somewhat opposite case, \emph{graphs with disjoint cycles}. A cycle is a path along the graph starting and ending at the same vertex, and it is simple if no other vertex is repeated. 
We say that $G$ has \emph{disjoint cycles} if distinct simple cycles do not share any vertex. See \S \ref{sec-disjoint} and Figure \ref{fig:disjoint}.
\begin{figure}
\includegraphics[height=2in,width=4in]{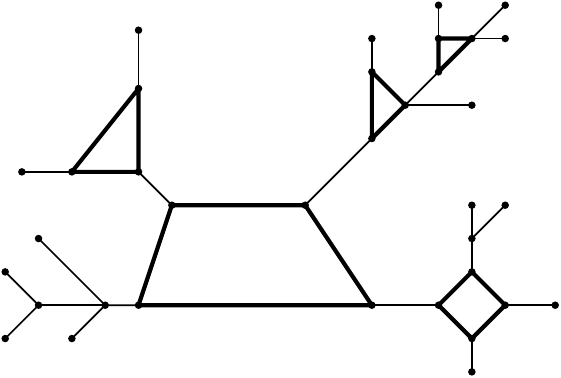}
    \caption{A graph with disjoint cycles}
    \label{fig:disjoint}
\end{figure}

If $\phi \in \RR^n$ is an eigenvector of $h \in \mathcal S(G)$
in order to avoid double subscripts we sometimes write $\phi(r)=\phi_r$.  To define the nodal count for all signing of $h \in \mathcal S(G)$, the matrix $h$ need to satisfy the following {\em generic} spectral condition:
\begin{enumerate}
\item [{[GSC]}]
We say $h \in \mathcal S(G)$ satisfies the {\em generic spectral condition}, abbreviated GSC, if $h$ is strictly supported on $G$, and every eigenvalue of every signing of $h$ is simple with nowhere vanishing eigenvector,
\end{enumerate}

In Lemma \ref{lem:genericity} we establish that condition [GSC] is indeed generic.
The main result of this paper is the following:

\begin{thm} \label{thm-intro} Let $G$ be a simple connected graph with $n$ vertices and disjoint cycles, let $h \in \mathcal S(G)$ that satisfy [GSC], and let \ct{$h'$ be a random signing of $h$}. \ct{Then for any $k\in[n]$, the random variable $\sigma(h',k)$ is binomial}\quash{the
nodal surplus distribution of $\lambda_k(h')$ is binomial, as $h'$ varies over the signings of $h$}:
the fraction of those signings $h'$ such that $\sigma(k,h ') = j$ is $2^{-\beta}\binom{\beta}{j}$.

Consequently, as $\beta\to\infty$, this distribution converges to a Gaussian centered at $\beta/2$ with variance $\beta/4$. \end{thm}

Our theorem was inspired by a related result (\cite[Theorem 2.3]{LiorFirstPaper}) for quantum graphs with disjoint cycles and $\QQ$-linearly independent edge lengths, where it was shown that the distribution of the nodal surplus for the first $N$ eigenvectors converges to a binomial distribution as $N\to\infty$ (A quantum graph has countably many eigenvalues). However, our case is different. We consider a fixed value of $k$ (the $k$-th eigenvalue) and the nodal count distribution over different signings of our operator (matrix). For example, $\beta$ may be much greater than the term $k-1$ in the nodal count.  (For
quantum graphs, on the other hand, $k$ grows to infinity while $\beta$ is fixed, so the nodal surplus is a small perturbation of the linearly growing nodal count.)

\subsection{}
Given a graph $G$ with a matrix $h$ as above, the various signings of $h$ lie in a single 
torus\footnote{defined (in \cite{NodalCount}) by allowing  off-diagonal elements $h_{rs}$  to vary by a phase $e^{i\theta_{rs}}h_{rs}$}
$\TT_h\subset \mathcal H(G)$, 
 see equation (\ref{eqn-TTh}).  We may consider the eigenvalue $\lambda_k$ to be
a sort of Morse function on $\TT_h$.  It is a theorem of Berkolaiko \cite{NodalCount}, further explained by Colin de
Verdi\`ere \cite{CdV} that each signing $h' \in \TT_h$ is a critical point of $\lambda_k$, whose Morse index
coincides with the nodal surplus for $h'$.  Unfortunately due to
the existence of a group of {\em gauge transformations} that acts on $\TT_h$ and preserves $\lambda_k$, 
each critical point $h'$ is highly degenerate.  

The degeneracy in the critical points can be removed by dividing the torus $\TT_h$ by the gauge group.  
The result is a torus $\mathcal M_h$ whose dimension 
\[\beta = |E|-n+1 ={\mathrm {rank}}(H_1(G))\]
is the first Betti number of $G$.
The genericity condition [GSC] now implies (\cite[Thm 3.2]{AG}) that for each signing $h'$ of $h$ the corresponding point
$[h'] \in \mathcal M_h$ is a nondegenerate critical point of $\lambda_k:\mathcal M_h \to \RR$.
One might hope, especially in the case of a graph with disjoint cycles, that these are the only 
critical points of $\lambda_{k}$. If this was the case then we would
conclude that $\lambda_k$ is a {\em perfect} Morse function, that each critical point contributes to the homology
of $\mathcal M_h$ in a single degree and hence the nodal surplus is binomially distributed.  A similar situation occurs in
 \cite[Theorem 3.2 and \S{3.4} ]{AG}, where it was proven that the nodal surplus distribution is binomial 
when $G$ is a complete graph and $h$ has a dominant diagonal. 
It is likely true, for generic graphs with disjoint cycles, that each $\lambda_k$ is a perfect Morse function on $\TT_h$, but we do not prove it.

\subsection{}   Instead, we develop a
 different approach  using the combinatorics of the Boolean lattice (\S \ref{subsec-outline}) 
and two technical steps: (a) the 
{\em monotonicity lemma} (Lemma \ref{lem-monotonicity}), and (b) the {\em local-global} theorem 
(Theorem \ref{local-global}).  
These results allow  us to focus on the one dimensional trajectories that connect neighboring 
signings $h'$ and ${h^{\prime\prime}}$ as described in Propositions \ref{prop-part1} and
\ref{prop-part2}.
The proof is then outlined in \S \ref{subsec-outline}.

\subsection{} An important ingredient in the proof is the {\em probability current} $\JJ(h, \phi)$
(Definition \ref{def-criticality}), a real anti-symmetric matrix supported on $G$, which may be interpreted as a gauge invariant
divergence-free vector field
or as a harmonic 1-form.  It is defined for any $h\in\mathcal{H}(G)$ and every eigenvector of $h$ and has a special structure. 
It vanishes on every bridge and is constant on the edges of each
simple separated cycle.  If the eigenvalue $\lambda$ is simple 
and the eigenvector is normalized then $-2\JJ$ is the derivative 
of $\lambda$, cf. Proposition \ref{prop-current}.

\subsection{Acknowledgements}  The authors are grateful to R. Band and G. Berkolaiko
for their encouragement.  Part of the proof
follows ideas developed by the first author together with  Band and
 Berkolaiko in a joint unpublished project studying a similar
question on quantum graphs.  We are particularly grateful to G. Berkolaiko for his
meticulous reading of an earlier version of this manuscript and for his many useful
comments and suggestions.  We would also like to thank Cynthia Vinzant for useful conversations. The first author is supported by the Simons Foundation Grant 601948, DJ.

\section{Recollections on graphs}
\subsection{}  As in \S \ref{sec-intro} we consider
a simple connected graph $G$ on $n$ ordered vertices numbered $1,2, \cdots, n$.  
 We write $\mathcal H_n, \mathcal S_n, \mathcal A_n$ for
the Hermitian, real symmetric, and real antisymmetric $n \times n$ matrices, and
we write $\mathcal H(G)$, $\mathcal S(G)$, $\mathcal A(G)$ for those  matrices supported on $G$.  
If $(rs)$ is an edge in $G$ write
$E[rs]$ for the matrix that is zero except for $E[rs]_{rs}=1$ and let
$A[rs] = E[rs]-E[sr]$ be the corresponding antisymmetric matrix.

%We say that a graph $(G,h)$ is {\em regular} if $h$ is strictly
%supported on $G$ and the eigenvalues of $h$ are simple with nowhere vanishing eigenvectors.

\subsection{}\label{subsec-recollections}
Each edge $(rs)$ of $G$ has a natural orientation ($+$ or $-$) which is the sign of $s-r$.
The space $C_1(G;\RR)$ of 1-chains consists of finite linear combinations of oriented edges
\[\gamma =
\sum_{\substack{{r \sim s}\\{r<s}}}\gamma_{rs}(rs) \ \text{ with }\ \gamma_{rs} \in \RR.\]
 The space $C_0(G,\RR)$ of 0-chains consists of formal finite linear
combinations $\sum_{r=1}^n a_r (r)$  of vertices with $\partial:C_1(G) \to C_0(G)$ defined by $\partial(rs) = (s) - (r)$.
We may consider $\mathcal A(G)$ to be the space of 1-forms $\Omega^1(G;\RR)$, dual
to $C_1(G)$ with respect to the bilinear pairing
\[  \int_{\gamma}\alpha := \sum_{\substack{{r \sim s}\\{r<s}}} \gamma_{rs}\alpha_{rs}\] 
where $\gamma \in C_1(G;\RR)$ and $\alpha \in \mathcal A(G)$.

The space of real valued functions defined on the vertices of $G$ is denoted $\Omega^0(G,\RR)\cong \RR^n$.
The differential $d:\Omega^0(G) \to \Omega^1(G)$ is 
\[(d\theta)_{rs} = \begin{cases} \theta(s) - \theta(r) &\text{if } r \sim s \\
0 &\text{otherwise}\end{cases}\]
Its adjoint with respect to the natural inner products\footnote{given by
$\langle \theta,\theta'\rangle = \sum_r \theta_r\theta'_r$ and
$\langle \alpha, \alpha' \rangle = \sum_{r<s}\alpha_{rs}\alpha'_{rs}$ for
$\theta \in \Omega^0$ and $\alpha \in \Omega^1$}
on $\Omega^0(G)$ and $\Omega^1(G)$ is
\[ (d^*\alpha)_r = \sum_s \alpha_{rs}.\]
\subsection{}
Stokes' theorem $\int_{\gamma} d\theta = \int_{\partial \gamma} \theta$ implies that the integration pairing passes to a 
nonsingular dual pairing between the cohomology 
$ H^1(G,\RR) = \ker(d)$ and  the homology $H_1(G,\RR) = \text{coker} (\partial)$.
%Then $\beta = \dim(H^1(G,\RR))$.
Consequently, {\em given $\alpha \in \mathcal A(G)$, there exists $\theta \in \Omega^0(G,\RR)$ such that
$\alpha = d\theta$ if and only if $\int_{\gamma}\alpha = 0$ for every cycle $\gamma$.} \qed

\subsection{Action of $\mathcal A(G)$}
The space $\mathcal A_n(\RR)$  of real $n \times n$ antisymmetric matrices acts on the space $\mathcal H_n$  by
\[ (\alpha *h)_{rs} = e^{i \alpha_{rs}h_{rs}}\]
with $\alpha'*\alpha*h = (\alpha'+\alpha)*h$. Let $\mathcal A_n(2\pi\ZZ)$ be the set
of antisymmetric matrices whose entries are integer multiples of $2\pi$.
The action factors through the torus $\mathcal A_n(\RR) / \mathcal A_n(2\pi\ZZ)$ so that
\[ \TT(G) = \left\{ \alpha \in \mathcal A_n(\RR)/\mathcal A_n(2\pi\ZZ):
\alpha_{rs} \ne0 \implies r \sim s \right\}\]
acts on $\mathcal H(G)$.  The mapping 
\[ *: \TT(G) \times \mathcal S(G) \to \mathcal H(G)\]
is a finite surjective covering.  For each $h \in \mathcal S(G)$  the orbit 
\begin{equation}\label{eqn-TTh}\TT_h = \TT(G)*h\end{equation}
is a torus of perturbations\footnote{referred to in \cite{NodalCount}
as the torus of ``magnetic perturbations of $h$" because, for the Schr\"odinger
operator, these perturbations arise from the introduction of a magnetic field, cf,~\cite{AG}.} of $h$.  
The torus $\TT_h$ is preserved under complex conjugation and the fixed points are the
intersection $\TT_h \cap \mathcal S(G)$, which consists of the  {\em signings} of $h$. 

\subsection{Gauge equivalence} \label{subsec-gauge}
If $\theta = (\theta_1,\theta_2,\cdots,\theta_n) \in 
\Omega^0(G,\RR)\cong \RR^n$
and $h \in \mathcal H(G)$ then
\[ d\theta *h = e^{i\theta}h e^{-i\theta}\]
is conjugate to $h$, where $e^{i\theta} = \text{diag}(e^{i\theta_1}, e^{i\theta_2}, \cdots, e^{i\theta_n})$.
Therefore $\lambda_k(d\theta *h) = \lambda_k(h)$.  If $V_{\lambda}(h) = \ker(h-\lambda I)$
then  
\begin{equation}\label{eqn-phase} 
V_{\lambda}(d\theta *h) = e^{i\theta}V_{\lambda}(h).\end{equation}
We say the elements $h$ and $h' = d\theta *h$  are {\em gauge equivalent} and differ by the {\em gauge transformation} $d\theta$.  Geometrically, equation (\ref{eqn-phase}) says that
eigenvectors $\phi, \phi'$ of $h$ and $h'$ differ by  changing the phases, $\phi'_r = e^{i \theta_r}\phi_r$.  Since their eigenvalues $\lambda_k, \lambda'_k$ are equal it makes sense to restrict attention to gauge-equivalence classes of matrices.

We may formally define the gauge group $\mathcal G = (\RR/2\pi\ZZ)^n$ with action 
 $\theta \diamond h = d\theta *h$,
whose orbits are gauge equivalence classes.  
The quotient of $\TT_h$ under gauge equivalence is an abstract torus 
$\mathcal M_h= \TT_h\sslash \mathcal G$, of
dimension $\beta$, the {\em manifold of magnetic perturbations modulo
gauge transformations}.  We sometimes write $[h] \in \mathcal M_h$ for the gauge-equivalence
class of $h$.  

Equation (\ref{eqn-phase}) reflects an action  of the
gauge group on vectors $\phi \in \CC^n$ with $\theta \diamond \phi = e^{i\theta}\phi$.

\section{Disjoint cycles}\label{sec-disjoint}
\subsection{}\label{subsec-separated1}  We say a graph $G$ has disjoint cycles if distinct simple cycles
do not share a vertex, cf. \S \ref{subsec-star}. Thus, each edge in $G$ is a bridge unless it is 
contained in a simple cycle.   Throughout this section we fix a graph $G$ with disjoint cycles 
and a matrix $h\in \mathcal S(G) $.  % which is generic in the sense of [GSC]. 
We also fix $k\in [n]=\{1,2,\cdots,n\}$  and consider the eigenvalue function $\lambda_k$. 
\subsection{The function $\Lambda_{k}$ and choice of basis for $\TT^{\beta}$}\label{spanning_tree} Fix a spanning tree in $G$.  Its complement consists of a single 
edge in each simple cycle.
%For each cycle $\gamma_j$ choose an edge $(r_j,s_j)\in \gamma_j$ such that $h_{r_js_j} \ne 0$.  
The elements $\alpha \in \TT(G)$
that are supported on these edges form a torus $\TT^{\beta}$ that projects isomorphically
to the quotient torus $\mathcal M_h$.  In other words, every element $\alpha *h \in \TT_h$ is
gauge equivalent to some $\alpha'*h$ where $\alpha'$ is supported on these chosen edges.
Thus, $\TT^{\beta}$ is a ``lift" to $\TT(G)$ of the manifold $\mathcal M_h$, as in the
following diagram.  The composition across the top row is denoted $\Lambda_k:\TT^{\beta} \to \RR$.

\begin{equation}\label{eqn-figT(G)}
	\begin{diagram}[size=2.5em]
		\TT^{\beta} & \rInto &\TT(G) & \rTo^{*h} & \TT_h & \rInto & \mathcal H_n & \rTo^{\lambda_k} & \RR&\\
	&\luTo(4,2)_{\cong}&	 && \dTo & && \ruTo(4,2)_{\lambda_k} &&\\
	&& & & \mathcal M_h &&&&&&
	\end{diagram}\end{equation}

\newcommand{\alphab}{\epsilon}

\subsection{Combinatorics of $\TT^{\beta}$}\label{subsec-combinatorics}  Choose an ordering of the
edges identified in \S \ref{spanning_tree} (with one edge in each simple cycle).  This
gives a particular choice of identification
\begin{equation}\label{eqn-torus-decomposition}
\begin{diagram}[size=2em]
\left(S^1\right)^{\beta} &\cong& \TT^{\beta} &\rTo^{\scriptstyle *h} & \mathcal M_h 
\end{diagram} 
\end{equation}  
Let $e_1,e_2,\cdots, e_{\beta} \in \TT^{\beta}$
denote the image in $\RR^{\beta}/(2\pi\ZZ)^{\beta}$ of the standard 
basis\footnote{each $e_j=A[r_js_j]=E[r_js_j]-E[s_jr_j]$ 
is in fact a matrix in $\mathcal A (G)$ defined modulo $2\pi$, and corresponds to one of the
particular edges identified in \S \ref{spanning_tree}. } vectors. 
Points  $\alphab=\sum_{i=1}^{\beta}\epsilon_i e_i \in \TT^{\beta}$ with coordinates  $\epsilon_i \in \left\{ 0, \pi \right\}$ are called symmetry points.  
By abuse of notation we write $\alphab \in \left\{0,\pi \right\}^{\beta}$.  The corresponding matrices $h_{\alphab}=\alphab*h$ are the signings of $h$ modulo gauge equivalence.  

\newcommand{\sym}{ \left\{ 0, \pi \right\}^{\beta}}
There are $2^{\beta}$ symmetry points in $\TT^{\beta}$.  
They form the vertices of a (hyper-)cube 
\[ \square \subset \mathcal M_h\]
whose $1$-skeleton consists of
edges that connect a symmetry point $\alphab$ to a neighbor $\alphab + \pi e_j \mod{2\pi}$ (where $j\in[\beta]$).  A choice of eigenvalue $\lambda_k$ determines a partial ordering on the symmetry points,
\[\alphab \succeq \alphab'\ \iff\ \lambda_k(\alphab *h) \ge \lambda_k(\alphab'*h).\]
  For $\alphab \in \sym \subset \TT^{\beta}$ let 
\begin{equation}\label{eqn-J}J_{-}(\alphab) =J_{-}(\epsilon, k, h) = \left\{ j \in [\beta]:\ \lambda_k((\alphab +\pi e_j)*h) < \lambda_k\left(\alphab*h\right)\right\}.
\end{equation}
The set $J_{-}(\alphab)$ identifies those neighbors $\alphab + \pi e_j$ of $\alphab$  in the 1-skeleton for which the 
eigenvalue $\lambda_k(h_{\alphab})$ decreases. 
%\lior{I'm not sure why we chose the $+$ subscript to describe the decreasing directions instead of a $-$ sign. Should we change that? Do we care?}

\quash{%%%%%  begin quash
\begin{thm}\label{thm-main} 
Let $G$ be a simple, connected graph with disjoint cycles, and let $h$ be a real symmetric matrix supported on $G$ 
that satisfies [GSC].  Fix $k$ $(1 \le k \le n)$.  Then,
as $h'$ varies over the signings of $h$ the nodal surplus distribution of $\lambda_k(h')$ is binomial, that is,
the fraction of those signings $h'$ such that $\sigma(k,h') = j$ is $2^{-\beta}\binom{\beta}{j}$.
\end{thm}

\subsection{Outline of proof}  \label{subsec-outline}
The proof has many technical steps but the ideas are relatively simple so we have organized the remainder of the
paper in a slightly unorthodox way, first stating the main steps in the proof and later providing their details.

}%%%% end quash

\subsection{}
Although the proof of our main result (Theorem \ref{thm-intro}) has many technical steps the ideas are
relatively simple, requiring only the following two propositions whose proofs appear in
\S \ref{sec-proofs}.
Let $G$ be a simple connected graph with disjoint cycles and suppose $h \in \mathcal S(G)$ 
is generic in the sense of [GSC]. 
Fix $k\in [n]$  and recall the notation $\Lambda_k(\alpha) =
\lambda_k(\alpha*h)$ for $\alpha \in \TT^{\beta}$.   

\begin{prop}\label{prop-part1}
  Each $\alphab \in \sym$ is a nondegenerate critical point of the function 
$\Lambda_k:\TT^{\beta} \to \RR$ and its Morse index is $\text{ind}(\Lambda_k)(\alphab) = |J_{-}(\alphab)|$.
 The Hessian of the function $\Lambda_k$ is diagonal with respect to the
decomposition (\ref{eqn-torus-decomposition}).
\end{prop}
\begin{prop}\label{prop-part2} The mapping $\left\{0,\pi \right\}^{\beta} \to \mathcal P{[\beta]}$ 
(the set of subsets of $[\beta]$), given by
$\alphab \mapsto J_-(\alphab)$ is bijective.  This implies that $\sym$ becomes a Boolean
lattice\footnote{The Boolean lattice on a finite set $S$ is the partially ordered set $\mathcal P(S)$
of subsets of $S$ ordered by inclusion.} under the above partial order. 
\end{prop}

\subsection{Proof of Theorem \ref{thm-intro}}\label{subsec-outline}
First we consider the nodal distribution of $\Lambda_k(\epsilon)$ as $\epsilon$ varies in 
$\{0, \pi\}^{\beta} \subset \TT^{\beta}$.
By \cite{NodalCount, CdV, AG}  the function $\Lambda_k$ has a nondegenerate critical point at each
$\alphab \in \sym$ and its Morse index equals the nodal surplus $\sigma(h,k)$  at that point.  By Proposition \ref{prop-part1}, this means that the nodal surplus distribution
coincides with the distribution of the numbers $|J_-(\alphab)|$.  Proposition \ref{prop-part2} implies that the distribution of the numbers $|J_-(\alphab)|$, and hence also the nodal surplus distribution for $\lambda_k$, is
binomial as $\alphab$ varies in $\sym$. 

Next we consider the set of signings of $h$.  The set $\sym*h$ is the quotient of the set of signings of 
$h$ by the action of the 
gauge group, or more accurately, the action by a certain subgroup of the gauge group.  
If $\theta = (\epsilon_1, \epsilon_2, \cdots, \epsilon_n) \in \Omega^0(G;\RR)$ 
with $\epsilon_i \in \{0,\pi\}$  and if $h' \in \mathcal A(G)$ is a signing of 
$h$ then $d\theta *h'$ is another signing.  The set of such $\theta$ form a group
under addition modulo $2\pi$.
If $h$ is properly supported on $G$ then this defines a free action of 
$(\ZZ/(2))^n$ on the set of signings (cf. \cite[\S 2.6, \S 2.7]{AG}).  
Each symmetry point $\alphab\in \sym \subset \TT^{\beta}$
corresponds to exactly the same number, $2^{n-\beta}$ of signings.
Therefore the binomial distribution on $\sym$ becomes the same binomial distribution on the set of signings. \qed

\section{Probability current and criticality}\label{sec-criticality}
Throughout this section we fix a simple connected graph $G$ with $n$ vertices and $h\in\mathcal{S(G)}$ strictly supported on $G$. %and when writing $h_{\alpha}$ we always mean that $\alpha\in\mathcal{A}(G)$ and $h_{\alpha}=\alpha*h$. 
\begin{defn} \label{def-criticality} Let $\alpha\in\mathcal{A}(G)$ and set $h_{\alpha} = \alpha*h$.
Given an eigenvector $\phi$ of $h_{\alpha}$,  define the {\em \textbf{probability current}} $\Cu=\Cu(h_{\alpha},\phi) \in \mathcal{A}(G)=\Omega^1(G,\RR)$ by
\[\Cu_{rs} = \Im\left((h_{\alpha})_{rs}\bar\phi_r \phi_s\right)=\Im\left(e^{i\alpha_{rs}}h_{rs}\bar\phi_r \phi_s\right).\]
We say that the eigenvector $\phi$ satisfies the
{\em \textbf{criticality condition}} at an edge $(rs)$ if $\Cu_{rs}=0$.\end{defn}
We remark that the probability current is defined for any eigenvector whether or not the eigenvalue is simple. 
\begin{prop}\label{prop-current}  
The probability current $\Cu=\Cu(h_{\alpha},\phi)$ satifies the following:\begin{enumerate}
    \item $\Cu$ is gauge-invariant, namely  $\Cu(d\theta*h_{\alpha},e^{i\theta}\phi)=\Cu(h_{\alpha},\phi)$. 
    \item $\Cu$ is divergence free, meaning that $d^* \Cu = 0$.
    \item \label{item-bridge}
    $\Cu_{rs}=0$ for every bridge\footnote{A bridge is an edge whose removal separates the graph $G$ into
    two pieces, each of which is a union of connected components} $(rs)$.
    \item $\Cu$ is constant along the edges of any simple cycle of $G$ that is disjoint from all others.
 
    \item \label{item-derivative} If $\lambda(h_{\alpha})$, the eigenvalue of $\phi$, is simple, then
    $\Cu$ is proportional to its derivative,
    \[ \frac{\partial \lambda(h_{\alpha})}{\partial \alpha_{rs}}= 
    \frac{\partial \Lambda}{\partial \alpha_{rs}} = -2\|\phi\|^{2}\JJ_{rs}.\]
\end{enumerate}
\end{prop}
We remark, in particular, if the criticality condition holds on an edge of a disjoint cycle then it holds on all the edges of that cycle. The proof of Proposition \ref{prop-current} will appear after 
a short review (\S \ref{subsec-derivative}) on derivatives of eigenvalues, which is used in the proof. 
\subsection{Derivatives of eigenvalues}  \label{subsec-derivative}
 Recall that $A{[rs]}$ is the antisymmetric matrix
with zero entries except for $A{[rs]}_{rs}=1$ and $A[rs]_{sr}=-1$.  Fix $\alpha \in \mathcal A(G)$, and consider the one-parameter family $\alpha(t)=\alpha+tA{[rs]}$ that goes through $\alpha$ in the $(rs)$ direction. The $t$-dependence of $\alpha(t)*h$ occurs  only in the $(rs)$ and $(sr)$ entries with 
\[(\alpha(t)*h)_{rs} = e^{it} e^{i \alpha_{rs}}h_{rs}
= e^{it}(h_{\alpha})_{rs}.\]
If $\lambda_k(h_{\alpha})$ is a simple eigenvalue then $t\mapsto \lambda_k(\alpha(t)*h)$ is an analytic function of $t$ around $t=0$, and its derivative at $t=0$ is the directional derivative of $\lambda_k(\alpha*h)$.

If $\lambda_k(h_{\alpha})$ has a nontrivial multiplicity then the function $\lambda_k(\alpha(t)*h)$ may fail to be differentiable.  The theorem of Kato (\cite[Thm.~1.8]{Kato}) and Rellich (\cite[Thm.~1]{Rellich}) implies that it is possible to find
analytic families of eigenvalues $\mu_k(t) \in \RR$ and eigenvectors $\phi_k(t)$, for all $t\in \RR$, so that $(\alpha(t)*h)\phi_k(t)
= \mu_k(t) \phi_k(t)$.  However the curves $\mu_k(t)$ may cross, when there are multiple eigenvalues, so the index $k$ does not necessarily
correspond to the order of these eigenvalues.  In other words, as $t$ varies, $\lambda_k(\alpha(t)*h)$ jumps between various analytic branches $\mu_j(\alpha(t)*h)$.  Let us choose one such analytic family or ``branch", 
$(\mu, \phi)$, and drop the subscript $k$, and define
\[ \Lambda: \TT^{\beta} \to \RR \ \text{ by }\ \Lambda(\alpha') = \mu(\alpha'*h).\]
%Its derivative, then, is
%\[\left.\frac{d}{dt} \mu\left((\alpha + tE[rs])*h  \right)\right\vert_{t=0}
%= \frac{\partial F}{\partial \alpha_{rs}}(\alpha)
%\]

Using Leibniz' dot notation to denote derivative with respect to $t$, and
differentiating $h(t)\phi = \mu(t) \phi(t)$ gives 
\begin{equation}\label{eqn-first-derivative}
(\dot h(t) - \dot \mu(t))\phi(t) + (h(t) - \mu(t))\dot \phi(t) = 0
\end{equation}
%\[\dot h(t) \phi(t) + h(t) \dot{\phi}(t) = 
%\dot{\mu}(t) \phi(t) + \mu(t) \dot{\phi}(t).\]
As in \cite[Lemma 2.5]{LocalTest} or 
\cite[\S 5.2]{AG}, taking the inner product with $\phi$ where $\Vert{\phi}\Vert = 1$, using that $h$ is Hermitian, and evaluating at $t = 0$  gives {\em the directional derivative of the eigenvalue $\mu$ along this branch:}
\begin{align}\begin{split}\label{eqn-partial-derivative}
\frac{\partial \Lambda}{\partial \alpha_{rs}}(\alpha) &=
\dot{\mu}= \left.\frac{d}{dt} \mu \left( (\alpha + t(A[rs])*h  \right)\right\vert_{t=0}=\langle \phi, \dot h\phi \rangle \\
& =i\left(\bar\phi_r \phi_s (h_{\alpha})_{rs} - \bar\phi_s \phi_r (\bar{h}_{\alpha})_{rs}\right) = -2\Im\left((h_{\alpha})_{rs} \bar\phi_r \phi_s\right)
\end{split}\end{align}
where $\phi_r = \phi(r)$ denotes the value of $\phi$ on the vertex $(r)$.

For later applications in equation \eqref{eqn-Omegajj}, consider the case when $\frac{\partial \Lambda}{\partial \alpha_{rs}}(\alpha)=0$. We can differentiate (\ref{eqn-first-derivative}) once again to obtain, as in \cite{LocalTest} Lemma 2.6,
\begin{equation}\label{eqn-ddoth}
\langle \phi, \ddot h \phi \rangle = 
-\Re\left((h_{\alpha})_{rs}\bar \phi_r \phi_s    \right)=-\left((h_{\alpha})_{rs}\bar \phi_r \phi_s    \right)
\end{equation}  and
\begin{equation} \label{ddot-mu} \frac{\partial^2\Lambda}{\partial \alpha_{rs}^2} = 
\ddot \mu = \langle \phi, \ddot h \phi \rangle + 
2 \Re\left(\langle \phi, \dot h \dot \phi \rangle\right).\end{equation}

\subsection{Proof of Proposition \ref{prop-current}}
The gauge invariance, $\Cu(e^{i\theta}h_{\alpha} e^{-i\theta}, e^{i\theta}\phi) = \Cu(h_{\alpha}, \phi)$ is straightforward from the definition.
The divergence is
\begin{align*}%%\label{eqn-current}
     (d^*\Cu)_r &= \sum_s\Im(\bar\phi_r(h_{\alpha})_{rs}\phi_s) = \Im\left(\bar\phi_r\sum_s(h_{\alpha})_{rs}\phi_s\right)\\
     &= \Im(\bar \phi_r \lambda \phi_r) =
     \lambda\Im(|\phi_r|^2) = 0.
 \end{align*}
    If removing an edge $E=(rs)$ separates the graph into two pieces, say $G_A$ and $G_B$, let $\theta \in \Omega^0(G)$
take the value $1$ on $G_B$ and $0$ on $G_A$.  Then $d\theta$ is supported on $E$ and
\[\Cu_{rs} = \langle d\theta, \Cu \rangle = \langle \theta, d^*\Cu \rangle = 0.\]
Similarly, if $E, E'$ are two edges in a simple cycle that is disjoint from all others then removing both
separates the graph into two pieces.  Taking $\theta$ as above, 
\[ \Cu(E) - \Cu(E') = \langle d\theta, \Cu\rangle = 0. \]
Part (\ref{item-derivative}) is a restatement of equation (\ref{eqn-partial-derivative}).
\qed

\begin{lem}{\rm(Partial criticality.)}\label{lem-conjugation}  
Let $\alpha\in\mathcal{A}(G)$ and set $h_{\alpha} = \alpha*h$. Let $\phi$ be an eigenvector of simple eigenvalue of $h_{\alpha}$ and let $\Cu=\Cu(h_{\alpha},\phi)$ be the probability current. Suppose there is a bridge that splits the graph $G$ into $G_A$ and $G_B$. If $h_{\alpha}$ is real on the $G_B \times G_B$ block, then $\Cu$ vanishes on that block,
\[h_{\alpha}|G_{B}\in\mathcal{S}(G_{B})\Rightarrow\JJ|G_B = 0.\]
\end{lem}

\begin{proof} Let $(rs)$ denote the bridge with $s\in G_{A}$ and $r\in G_{B}$. By changing gauge and scaling $h_{\alpha}$ if needed, we can assume that $(h_{\alpha})_{rs}=1$. Let
$e_{s}$ and $e_{r}$ be the corresponding standard basis vectors so that in the block decomposition to $G_{A},G_{B}$ we write $h_{\alpha}=A\oplus B +e_{r}e_{s}^{*}+e_{s}e_{r}^{*}$. Suppose the simple eigenvalue of interest is $\lambda=0$ (otherwise replace $h_{\alpha}$ with $h_{\alpha}-\lambda I$), and let $
\phi=(\phi_{A},\phi_{B})$ denote its normalized eigenvector. We need to show that if $B$ is real then $\phi_{B}$ is (proportional to) a real vector, in which case $\Cu|G_{B}=0$. If $\phi_{A}=0$ then $\phi_{B}\in\ker(B)$ and we are done. So assume $\phi_{A}\ne0$. 
By Proposition \ref{prop-current} (\ref{item-bridge}) we know that $\mathrm{Im}[(h_{\alpha})_{rs}\bar{\phi}_{r}\phi_{s}]=\mathrm{Im}[\bar{\phi}_{r}\phi_{s}]=0$, so by scaling $\phi$ if needed we can assume that $\phi(r)$ and $\phi(s)$ are real. We will now show that $\phi'=(\phi'_{A},\phi'_{B}):=(\phi_{A},\overline{\phi}_{B})$ is also in $\ker(h_{\alpha})$, which means that $\phi'=\phi$ since the kernel is one-dimensional and $\phi_{A}\ne 0$. This is true because 
\[h_{\alpha}\phi=0\Rightarrow A\phi_{A}+\phi(r)e_{s}=0\ \text{ and }\ B\phi_{B}+\phi(s)e_{r}=0,\]
and $\phi(r)=\phi'(r),\phi(s)=\phi'(s)$, so 
\begin{align*}
    \left(A\oplus B +e_{r}e_{s}^{*}+e_{s}e_{r}^{*}\right)\phi'= & (A\phi_{A}+\phi(r)e_{s},B\bar{\phi}_{B}+\phi(s)e_{r})\\
    = & (0,B\bar{\phi}_{B}+\phi(s)e_{r}),
\end{align*}
and since $B$ and $\phi(s)$ are real, then $B\bar{\phi}_{B}+\phi(s)e_{r}=\overline{(B\phi_{B}+\phi(s)e_{r})}=0$.\qedhere

\end{proof}

We now return to the special case of $G$ that has disjoint cycles. Recall that  
each $\epsilon
\in \{0,\pi\}^{\beta} $ is a non-degenerate critical point of $\Lambda_{k}: \TT^{\beta} \to \RR$.

%and consider the Hessians of each $\Lambda_k:\TT^{\beta} \to \RR$, with $\Lambda_k(\alpha) =
%\lambda_k(\alpha*h)$, at the different signings $\epsilon
%in \{0,\pi\}^{\beta} \subset\TT^{\beta}$, in the basis described in \S \ref{spanning_tree} (with one edge in each simple cycle).
\begin{cor} \label{cor-diagonal}  Suppose $G$ has disjoint cycles and $h \in \mathcal S(G)$ satisfies [GSC], then for each $k$, the Hessian of $\Lambda_{k}$ at any $\epsilon
\in \{0,\pi\}^{\beta} \subset\TT^{\beta}$ is diagonal with respect to the basis of $\TT^{\beta}$
that was chosen in \S \ref{spanning_tree}.
\end{cor}%%{Proof of Proposition \ref{prop-part2}}

\begin{proof} 
%Each point $\epsilon \in \sym$ is a nondegenerate critical
 % point of $\Lambda_k$ by \cite{NodalCount, CdV, AG} as we assumed $h$ satisfies [GSC].  We will show that its Hessian
  %is diagonal by verifying that the second mixed partial
  %derivatives vanish.  
Fix $k$ and $\epsilon$. We work in the previously chosen (\S \ref{spanning_tree}) basis of $T_{\epsilon}\TT^{\beta}\cong\RR^{\beta}$ given by the choice of a single edge per cycle of $G$, say $(r_j,s_j) \in \gamma_j$. We will show that 
  \begin{equation}\label{eqn-second-partial}
  \frac{\partial^2\Lambda_k}{\partial \alpha_{r_1s_1} \partial \alpha_{r_2s_2}}(\epsilon) = 0.  \end{equation}
(All other off-diagonal terms vanish for the same reason). Since the cycles are disjoint there exists a bridge that separates the graph into two parts, $G_A, G_B$ with $\gamma_1 \subset G_A$ and $\gamma_2 \subset G_B$.
Let $\alpha(t) = \epsilon + tA[r_1,s_1]$ and let
$h_t = \alpha(t)*h$.  The matrix $h_t$ is real except for the $(r_1,s_1)$ and
$(s_1,r_1)$ entries so we may apply Lemma \ref{lem-conjugation} and part (5) of Proposition \ref{prop-current} to
conclude that
\[ \frac{\partial \Lambda_{k}}{\partial \alpha_{r_2,s_2}}(\alpha(t)) = 0
\]
for all $t$ around $0$.  Differentiating with respect to $t$, at $t=0$, gives
equation (\ref{eqn-second-partial}).\qedhere \end{proof}

\section{Montonicity}\label{sec-montonicity}

\begin{lem} $(${\bf Montonicity}$)$\label{lem-monotonicity} Suppose $G$ has a cycle $\gamma$ disjoint from all others and let $h\in\mathcal{S}(G)$ that satisfies [GSC]. Consider any one-parameter family $h_{t}=\alpha_{t}*h$ with  $\alpha_{t}$ supported on $\gamma$ and $\int_{\gamma}\alpha_{t} = t$ for all $t\in[0,\pi]$.   Then, $t\mapsto\lambda_{k}(h_{t})$ is strictly monotone in $t\in [0,\pi]$ for all $k\in [n]$. 
\end{lem}
We remark that, up to gauge equivalence, we may suppose that $\alpha$ is supported on a single edge of $\gamma$.   
This means the family $h_t $ traverses a single segment in the 1-skeleton of the 
hypercube $\square \subset \mathcal M_h$ from \S \ref{subsec-combinatorics}.  The monotonicity lemma {\em only}
applies to these special paths.
The rest of \S \ref{sec-montonicity} is devoted to the proof Lemma \ref{lem-monotonicity}, which appears
finally in \S \ref{subsec-proof-montonicity}.

\begin{lem}[Flat band criteria]\label{lem-cycle criticality}
Suppose $G$ has a cycle $\gamma$ disjoint from all others, and $(12)$ is an edge in $\gamma$. Let $h\in\mathcal{S}(G)$ and consider a one-parameter 
family $h_{t}=\alpha_{t}*h$ {\color{black}{
where $\alpha_{t}\in\mathcal{A}(G)$ satisfies $\alpha_{t}=\alpha_{0}$ outside of $\gamma$ and $\int_{\gamma}\alpha_t = t$  for all $t\in [0,\pi]$.}} %supported on $\gamma$ and $\int_{\gamma}\alpha_{t} = t$ for all $t\in[0,\pi]$. 
Suppose there exists $t_{0}\in(0,\pi)$, and an eigenvector $\phi$ of $h_{t_{0}}$ with eigenvalue $\lambda$, such that $\Cu(h_{t_{0}},\phi)_{12}=0$. Then $\lambda$ is a common eigenvalue of all $h_{t}$ with $t\in[0,\pi]$.  
\end{lem}
\quash{
%%%%%%%%%%%%%%%%%%%%%%%%%%%%%%%%%%%%%%%%%%%%%
\begin{remark}[Intuitive explanation of the proof]
     Heuristically, the proof has two steps The proof has two steps, one step First note that $\Cu|\gamma=0$ since $\Cu$ is constant on $\gamma$, which means that  we  $\Cu_{rs}=0$ condition means that $\Cu|\gamma=0$ and therefore  vanish on all 
     An intuitive explanation for the proof is that having t $\gamma\cap\mathrm{support}(\phi)$ is a union of paths, and on each such path, $\phi$ has a constant phase. These phases add up to $t_{0}\mod{\pi}$, which we show implies the existence of a vertex of high degree on which $\phi$ vanishes. By a choice of Gauge the $t$ dependence can appear only at the equation at that vertex, and by adjusting the phases and properly scaling the parts of $\phi$, we create a $\lambda$ eigenvector for any $t$. 
\end{remark}
%%%%%%%%%%%%%%%%%%%%%%%%%%%%%%%%%%%%%%%%%%
}
\begin{proof}
\ct{Without loss of generality assume that $\lambda=0$, so that  $h_{t_{0}}\phi=0$. We need to provide a family of vectors $\phi_{t}$ such that $h_{t}\phi_{t}=0$ for all $t$. We will show that} {\color{black}{ $\JJ_{12}=0$ implies }} \ct{either:
\begin{enumerate}
    \item[(i)] there is an edge $(rs)$ in $\gamma$ such that $\phi(r)=0$ and $\phi(s)=0$, or
    \item[(ii)] there is a vertex $r$ in $\gamma$ such that $\phi(r)=0$ and $\deg(r)\ge 3$.
\end{enumerate}
We will also show that each of these conditions is sufficient for constructing $\phi_{t}$ such that $h_{t}\phi_{t}=0$ for all $t$.}\\
To ease notation let $\alpha=\alpha_{t_{0}}$. To avoid triple subscripts write $\alpha(rs)$ for $\alpha_{rs}$. 

\ct{First we show that (i) is sufficient.}
Let $(rs)$ be an edge in $\gamma$
%.  There is a gauge equivalent
%$\alpha' = \alpha +d\theta$ so that %$\alpha'_{pq}=0$ for every other edge $(pq) \in %\gamma$,  $(pq) \ne (rs)$,
%because $\gamma$ is a simple cycle.  If
\ct{such that} $\phi(r)=\phi(s)=0$. \ct{Up to gauge equivalence, we may assume that  $\alpha_{t}=\alpha$ on $G\setminus\gamma$,
that $\alpha_{t}(rs)=t$, and that $\alpha_t$ vanishes on all the other edges in $\gamma$. } 
Then \ct{ $h_{t}\phi=h_{t_{0}}\phi=0$ for all $t$ so we may take $\phi_t = \phi$.} 

%is independent of $t$, and is 
%$h_{t_{0}}\phi = 0$ implies that
%$h_{\alpha'_t} -\lambda) \phi = 0$ for all $t$, and the result follows.

\ct{Next, we show, using $\Cu_{12}=0$, that if (i) fails then (ii) must hold.} Assume (i) fails, namely
%Thus we may assume, 
\begin{enumerate}
    \item[(A)]for every edge $(rs)$ in $\gamma$, $\phi(r)$ and $ \phi(s)$ are not both zero.
\end{enumerate}
%Next we claim 
%\begin{enumerate}
 %   \item[($\dag$)] there exists a vertex $r\in\gamma$ of degree $\deg(r)\ge 3$ such that $\phi(r)=0$.
%\end{enumerate}
{\color{black}{By proposition \ref{prop-current}, $\Cu=\Cu(h_{t_{0},\phi})$ is gauge invariant and constant on $\gamma$, so $\Cu_{rs}=\Cu_{12}=0$ for every edge $(rs)$ in $\gamma$.}} To prove that (ii) holds, up to change of gauge we may assume that $\phi$ is real, cf. equation
(\ref{eqn-phase}).  In this case,  $\Cu_{rs}=\phi(r)\phi(s)h_{rs}\sin(\alpha(rs))=0$ for any $(rs)\in\gamma$, so 
{\color{black}{$\alpha(rs)=0\mod{\pi}$ when $\phi(r)\phi(s)\ne0$, and so (A) gives \begin{equation}\label{eqn-sum all} \sum_{\phi(r) = 0} \alpha(sr) + \alpha(rt) = \int_{\gamma}\alpha_{t_{0}}=t_{0}\ne 0\mod{\pi},\end{equation}
where $s<r<t$ denote the neighbors of $r$.}}
%\begin{enumerate}
%\item[(B)  ] If $(rs) \in \gamma$ and %$\alpha(rs) \notin \{0,\pi\}$ then
%$\phi(r)\phi(s)=0$.    
%\end{enumerate}
Now suppose $r$ is a vertex in $\gamma$ of degree 2 with $\phi(r) = 0$. Let $s < r < t$ be the two vertices attached
to $r$.  Then $\phi(s) \ne 0$ and $\phi(t) \ne 0$ by (A) so {\color{black}{$(h_{t_{0}}\phi)_{r}=0$}}\quash{the eigenvalue equation} reads
\[
\phi(s)h_{sr}e^{i\alpha(sr)} + 0 + \phi(t) h_{rt} e^{i\alpha(rt)} = 0
\]
which implies
\[\alpha(sr) + \alpha(rt) = 0 \mod \pi\]
whenever $\deg(r) = 2$ and $\phi(r) = 0$.
Adding over all vertices $r$ of degree 2, such that $\phi(r) = 0$ gives
\begin{equation}\label{eqn-sum} \sum_{\phi(r) = 0\ct{,\deg(r)=2}} \alpha(sr) + \alpha(rt) = 0\ct{\mod{\pi}}\end{equation}
where, as before, $s<r<t$ denote the neighbors of $r$.
The terms in this sum are disjoint by (A). {\color{black}{Since the sums in \eqref{eqn-sum} and \eqref{eqn-sum all} are not equal, then there must be a vertex $r\in\gamma$ with $\deg(r)\ne 2$ and $\phi(r)=0$, so (ii) holds.}} %If \ct{condition (ii) fails, namely} all vertices $r$ with $\phi(r)=0$
%have degree 2, then every edge $(pq)$ with $\alpha(pq)\ne 0\mod\pi$ 
%occurs in the sum, by (B). So the sum \eqref{eqn-sum} \ct{equals }
%actually computes 
%\[\ct{\int_{\gamma}\alpha  =t_{0}\ne 0\mod{\pi}}\] which 
%contradicts the assumption.  {\color{blue}{which assumption?  Sorry, there are 6 assumptions above and I lost track.}} 
%\ct{ Hence, condition (ii) holds.}
%that $\int_{\gamma}\alpha \ne 0\mod{\pi}$.  This proves ($\dag$).

\ct{Finally, assuming (ii) we construct $\phi_{t}$.} %Therefore, 
Without loss of generality, we may suppose that we have \ct{consecutive} vertices
$1<2<3$ in $\gamma$ with $\phi(2)=0,\deg(2)\ge 3.$  \ct{Up to gauge equivalence, we may assume that $(\alpha_{t})_{12}=t$,} {\color{black}{$(\alpha_{t})_{rs}=0$ for all other edges $(rs)$}}
%is zero on all other edges 
\ct{in $\gamma$, and $\alpha_{t}=\alpha$ on $G\setminus\gamma$.}

%$\alpha(12)+\alpha(23)\ne 0\mod{\pi}$. In particular, 
%\[\phi(1)h_{12}e^{i\alpha(12)}+\phi(3)h_{23}e^{i\alpha(23)}\notin \RR.\]
Let $H$ denote the union of connected components of $G\setminus\gamma$ that are connected to vertex $2$ in $G$.  We will show there exists
$c(t) \in \CC$ \ct{with $c(t_{0})=1$} so that the vector $\phi_{\ct{t}}$ defined by
  \[(\phi_{\ct{t}})_{r}=\begin{cases} c\ct{(t)} \phi_{r} & \text{for}\ r\in H\\
        \phi_{r} & \text{for}\  r\notin H
    \end{cases} \]
\ct{satisfies $h_{t}\phi_{t}=0$ for all $t$.}
%$\alpha_t*h$ 
%with fixed eigenvalue $\lambda$.
%Up to gauge equivalence, we may assume that $\alpha_t = \alpha$ on all edges except for (12) and that $\alpha_t(12) = t$.  
By our assumption on $G$, the only vertex in $G\setminus H$ with a neighbor in $H$ is the vertex $2$ on which $\phi(2)=0$. Therefore \ct{$(h_{t}\phi_{t})_{r}\propto(h_{t_{0}}\phi)_{r}=0$}
%$\left((h_{\alpha_{t}}-\lambda)\phi_{c}\right)(r)=0$
for all $r\ne 2$.  To 
understand the situation at vertex $2$ let
\[F_{2}(\alpha,\phi)=\sum_{r\in H, r\sim 2}h_{2r}e^{i\alpha({2r})}\phi(r).\]
Then $F_{2}(\alpha_{t},\phi_{\ct{t}})=c\ct{(t)}F_{2}(\alpha,\phi) $ and at vertex 2 we have
\[\ct{(h_{t_{0}}\phi)_{2}=}\quash{\left((h_{\alpha}-\lambda)\phi\right)(2)=}\phi(1)h_{12}e^{i\ct{t_{o}}}+\phi(3)h_{23}\quash{e^{i\alpha(23)}}+F_{2}(\alpha,\phi)=0.\]
Therefore, 
\[-F_{2}(\alpha,\phi)=\quash{\left((h_{\alpha}-\lambda)\phi\right)(2)=}\phi(1)h_{12}e^{i\ct{t_{o}}}+\phi(3)h_{23}\quash{e^{i\alpha(23)}}\ne 0\] since it is not real. The eigenvalue equation at vertex 2 becomes

\begin{align*}
  \ct{(h_{t}\phi_{t})_{2}=}\quash{\left((h_{\alpha_{t}}-\lambda)\phi_{c}\right)(2)=} & \phi(1)h_{12}e^{it}+\phi(3)h_{23}+F_{2}(\alpha_{t},\phi_{c})\\
  = & \phi(1)h_{12}e^{it}+\phi(3)h_{23}+c\ct{(t)} F_{2}(\alpha,\phi).
\end{align*}
\ct{It is left to choose $c(t)$ so that $(h_{t}\phi_{t})_{2}=0$, namely }
%We conclude that $\lambda$ is an eigenvalue of $h_{\alpha_t}$ with $(h_{\alpha_{t}}-\lambda)\phi_{c}=0$ by choosing
\begin{equation}
  c\ct{(t)}=-\frac{\phi(1)h_{12}e^{it}+\phi(3)h_{23}}{F_{2}(\alpha,\phi)}=\frac{\phi(1)h_{12}e^{it}+\phi(3)h_{23}}{\phi(1)h_{12}e^{it_{0}}+\phi(3)h_{23}}.  \qedhere
\end{equation}

\end{proof}
\quash{
%%%%%%%%%%%%%%%%%%%%%%
\subsection{}  \label{subsec-consequence}
Suppose $\gamma$ is a cycle disjoint from all others in $G$ and
suppose $\alpha \in \mathcal A(G)$ is supported on a single edge of $\gamma$, say $\alpha = A[12]$ in
the notation of \S \ref{subsec-derivative}.  Abbreviate $h_{t}=(t\alpha)*h$ so that $(h_{t})_{ij}=h_{ij}$ 
for all $ij$ except for
$(h_{t})_{12}=h_{12}e^{it}$ and $(h_{t})_{21}=h_{12}e^{-it}$.
%%%%%%%%%%%%%%%%%%%%%%%%%%%%%%%%%%%%%
}
\begin{lem}\label{lem-Jnot0} In the setting of Lemma \ref{lem-monotonicity}, if $(rs)$ is an edge in $\gamma$ and $\phi$ is a normalized eigenvector of $h_{t}$, for some $t\in(0,\pi)$, then $\Cu=\Cu(h_{t},\phi)$ has $\Cu_{rs}\ne0$. 

In particular, if $\phi$ and $\phi'$ are eigenvectors of the same eigenvalue of $h_{t}$, then 
$\Cu(h_t, \phi)_{rs}$ and $\Cu(h_t, \phi')_{rs}$ share the same sign.
\end{lem}
\begin{proof}
Since $h$ satisfies [GSC] then each of the eigenvalues has $\Lambda_{k}(\alpha)=\lambda_{k}(\alpha*h)$ has a non-degenerate critical point at $\alpha=0$, namely at $h=h_{0}$, whose Hessian is diagonal by Corollary \ref{cor-diagonal}. In particular, for any $k\in[n]$, $\Lambda_{k}(\alpha_{t})=\lambda_{k}(h_{t})$ is not constant around $t=0$. This means that $\Cu_{rs}\ne0$ for any normalized eigenvector of any $h_{t}$ with $t\in(0,\pi)$, otherwise we would get a ``flat band'', namely a constant eigenvalue $\lambda_{k}(h_{t})\equiv \lambda$ for all $t$ around $t=0$ by Lemma \ref{lem-cycle criticality}. This concludes the first part.

Now let $V=\ker(h_{t}-\lambda_{k}(h_{t}))$ be some eigenspace of some $h_{t}$ with $t\in(0,\pi)$, and assume $\dim(V)\ge 2$. Then the map $\phi\mapsto \Cu(h_{t},\phi)$ is a continuous map from $V\setminus\{0\}$ (which is connected) to $\RR\setminus\{0\} $ so its image must lie either in $\RR_{>0}$ or in $\RR_{<}$. $V\setminus \{0\}$.
\end{proof}   
     
 \subsection{Proof of Lemma \ref{lem-monotonicity}}\label{subsec-proof-montonicity}
 The statement is gauge invariant, so we may fix the gauge such that $\alpha$ is supported on a single edge, say, $(12)$. 
By Kato \cite[Thm 1.8]{Kato} or Rellich \cite[Thm. 1]{Rellich}, since this is a one-parameter analytic family of hermitian matrices, 
the ordered eigenvalues $(\lambda_1 \le \cdots \le \lambda_n)$
and eigenvectors $(\phi_1, \cdots, \phi_n)$ of $h$ extend analytically to 
eigenvalues and normalized eigenvectors $(\mu_{k}(t),\phi_{k}(t))_{k=1}^{n}$ of $h_{t}$, although apriori their order may not be preserved.
The derivative 
\begin{equation}\label{eqn-dotmu}
\dot \mu_k(t) = \frac{d}{dt}\mu_k(t)= \langle \phi_k(t), \dot h_t \phi_k(t)\rangle = -2\Cu(h_t, \phi_k(t))_{12}\end{equation} 
was calculated in (\ref{eqn-partial-derivative}). Since $\Cu(h_t, \phi_k(t))_{12}\ne0$ for all $k$ and all $t\in(0,\pi)$ by Lemma \ref{lem-Jnot0}, then each $\mu_{k}(t)$ is strictly monotone in $t\in[0,\pi]$.  If all eigenvalues are simple, this proves that $\lambda_k(h_t)=\mu_{k}(t)$ is monotone
for $t \in [0,\pi]$.

If the eigenvalue has a nontrivial multiplicity, say $\mu_k(t) =\mu_{k'}(t)$ then it suffices to know that the derivatives $-2\Cu(h_t, \phi_k(t))_{12}$ and $-2\Cu(h_t, \phi_k'(t))_{12}$ have the same signs. The second part of Lemma \ref{lem-Jnot0} ensures this is the case. \qed

\subsection{Remark} Lemma \ref{lem-Jnot0} and equation \eqref{eqn-dotmu} mean 
that the restriction of the Hermitian form $\dot h_t$ to the eigenspace of $h_t$ is sign-definite,
which is exactly the condition of \cite{GenericFamilies} for a point of multiplicity to be topologically regular (the BZ condition), see Appendix \ref{sec-BZ}.

\section{Genericity}\label{sec-genericity}
\begin{lem}\label{lem:genericity}
Let $G = G([n],E)$ be a finite simple connected graph.  The set of matrices 
\[ \mathcal O = \left\{h \in \mathcal S(G):\ h \text{ satisfies [GSC]  } \right\}\]
is open and dense in $\mathcal S(G)$.  Its
complement is contained in a closed semi-algebraic\footnote{A semi-algebraic subset of a real vector space is a finite Boolean combination of sets defined by
polynomial equalities $f(x) = 0$ and inequalities $f(x)<0$.} subset of $\mathcal S(G)$ of codimension $\ge 1$.
\end{lem}
\begin{proof}  When eigenvalues are simple, the eigenvalues and eigenvectors vary continuously with the matrix, so the set of matrices satisfying [GSC] is open.
If a matrix $h' \in \mathcal S(G)$ fails to satisfy [GSC] then there is a signing $\epsilon \in \{0, \pi\}^E$ for which $h = \epsilon *h'$ lies in at least one of the following sets:
\begin{enumerate}
    \item[(i)] The set of matrices in $\mathcal S(G)$ that are not strictly supported on $G$, that is, $h_{ij}=0 $ for some $(ij)\in E$.
    \item[(iii)] The set of matrices in $\mathcal S(G)$ that have a multiple eigenvalue: $\mathrm{discriminant}(h)=0$, or
    \item[(iii)] The set of matrices in $\mathcal S(G)$ that have a simple eigenvalue with an eigenvector that vanishes at some vertex. 
\end{enumerate}
The sets (i) and (ii) are zero sets of polynomials that are not the zero polynomial\footnote{To see that $\mathrm{discriminant}(h)\ne0$ for some $h\in\mathcal S(G)$ take $h=\Diag(1,2,\ldots,n)$. } on $\mathcal S(G)$, so these are algebraic subsets of $\mathcal S(G)$ with positive codimension in $\mathcal S(G)$.  Since the class of semi-algebraic subsets is
preserved by any change of signing, it suffices to show that the set (iii) is semi-algebraic with positive codimension in $\mathcal S(G)$.

Let us consider the set of $h \in \mathcal S(G)$ which admit an eigenvector $\phi = (0, \phi') \in \RR^n$ that vanishes on the first vertex.  
We may write $h$ in block form
\[ h = \left( \begin {matrix} a & b^t \\ b & d \end{matrix} \right)\]
with $a = h_{11}\in A=\RR$, with $b \in B = \RR^{n-1}$ and $d \in \mathcal S(G')$ where $G'$ is the graph obtained from $G$ by removing vertex no. 1.
Then $(h-\lambda.I)\phi = 0$ is equivalent to
\[ (d-\lambda.I)\phi' = 0\ \text{ and }b^t.\phi' = 0.\]
There is a diagram of algebraic sets
\begin{diagram}[size=2em]
&\widetilde{\mathcal H} \times A & \rTo & \widetilde{\mathcal H} & \rTo & 
\mathcal H & \rTo & \mathcal S(G') \\
&\dTo & & & & & & \\
&\mathcal S(G)
\end{diagram}
Here,
\begin{align*}
    \widetilde{\mathcal H} &= \left\{(d, \lambda, [\phi'], b)\in \mathcal S(G') \times \RR \times \PP(\RR^{n-1}) \times B:\ (d-\lambda.I)\phi' = 0 \text{ and } b^t.\phi' = 0\right\} \\
    \mathcal H &= \left\{ (d, \lambda, [\phi'])\in \mathcal S(G')\times \RR \times \PP(\RR^{n-1}):\ (d-\lambda.I)\phi' = 0 \right\}
\end{align*}
and we have written $[\phi'] \in \PP(\RR^{n-1})$ for the line defined by $\phi' \in \RR^{n-1}$.
We claim that
\[\dim(\widetilde{\mathcal H} \times A) = \frac{n(n+1)}{2}-1.\] 
This holds because: \begin{itemize}
    \item $\dim(\mathcal H) = \frac{(n-1)n}{2}$ since $\mathcal H \to \mathcal S(G')$ is generically $(n-1)$ to one. 
    (A generic element $d \in \mathcal S(G')$ has distinct eigenvalues 
$\lambda$, each with a unique projective eigenvector $[\phi']$).
\item $\dim(\widetilde{\mathcal H}) = \frac{(n-1)n}{2} + n-2 = \frac{n(n+1)}{2}-2$ since  
$\widetilde{\mathcal H} \to \mathcal H$ is a vector bundle whose fiber over $(d, \lambda, [\phi'])$
is the $n-2$ dimensional vector space $V_{[\phi']} = \left\{ b \in B:\ b^t.\phi' = 0\right\}$ (which depends only on the line $[\phi']$).    
\end{itemize}
The mapping $ \widetilde{H} \times A \to \mathcal S(G')$ is
\[ (d, \lambda, [\phi'],b,a) \to \left( \begin{matrix} a & b^t \\ b & d \end{matrix} \right).\]
Its image has dimension $\le \frac{n(n+1)}{2} -1$ and is exactly the set of $h \in \mathcal S(G)$ which have an eigenvector $\phi = (0,\phi')$ whose first coordinate 
vanishes.  By the Tarski-Seidenberg theorem, it is semi-algebraic.

Applying this argument to each coordinate gives $n$ such semi-algebraic sets, whose union is therefore also semi-algebraic of codimension $\ge 1$.
\end{proof}

\subsection{}
Let $G$ be a finite graph.
Let $\mathcal B$ denote the set of matrices $h\in \mathcal S(G)$ that satisfy 
\begin{itemize}
    \item[(\maltese)] 
any two gauge-inequivalent signings $\epsilon*h, \epsilon'*h$ have distinct eigenvalues.
\end{itemize}
\begin{lem} \label{lem-dag}
The set $\mathcal B$ is open and dense in $\mathcal S(G)$ 
and its complement is contained in an algebraic subset of codimension $\ge 1$.
\end{lem}  
\begin{proof}    
 We may assume that $G$ is connected.  If $G$ is a tree and $h \in \mathcal S(G)$
then every signing of $h$ is gauge equivalence to $h$ so we may assume $\beta(G) \ge 1$.  
First consider the case that $\beta(G) =1$ so that $G$ contains a unique cycle.  Fix an edge $(rs)$ in this cycle.  For
any $h \in \mathcal S(G)$ there is only one
gauge-equivalence class of signings $\epsilon*h$ of $h$ and it corresponds to changing the sign of $h_{rs}$ (and of $h_{sr}$).

Let $Q_{\epsilon}(h)$ denote the discriminant of the $2n\times 2n$ matrix $(h)\oplus(\epsilon*h)$.  
The set $Q_{\epsilon}^{-1}(0)$ is an algebraic subset of $\mathcal S(G)$ which contains the complement of $\mathcal B$.  
If $Q_{\epsilon}^{-1}(0)$  contains an open subset
of $\mathcal S(G)$ then it is all of $\mathcal S(G)$;  otherwise it has codimension $\ge 1$.
We will assume that $Q_{\epsilon}(h)=0$ for all $h \in \mathcal S(G)$ and arrive at a contradiction.

In this one dimensional case the hypercube of \S \ref{subsec-combinatorics}  is just an interval whose endpoints 
are $h$ and $\epsilon*h$.  Let $V={\rm diag}(1,2,\ldots,n)$. Let $\xi\in \mathcal S(G)$ (strictly supported on $G$)
sufficiently small such that $h:=V+\xi \in \mathcal O$ 
and $\epsilon*h \in \mathcal O$ (such $\xi$ exists by Lemma \ref{lem:genericity}).  The eigenvalues of $h$ are distinct; the eigenvalues of
$\epsilon*h$ are distinct.  Therefore, if $Q_{\epsilon}(h) = 0$
then $h$ and $\epsilon*h$ share an eigenvalue, say,  $\lambda_k(h) = \lambda_{k'}(\epsilon*h)$ .  If $\xi$ is
sufficiently small, the eigenvalues of $h$ and of $\epsilon*h$ are small perturbations of the eigenvalues
of $V$, which are distinct integers, hence $k = k'$.
But this contradicts the montonicity Lemma \ref{lem-monotonicity}.

We conclude that for any graph $G$ with $\beta(G)=1$ the function $Q_{\epsilon}$ vanishes identically on $\mathcal S(G)$.
Now consider the case of a general graph $\beta(G)\ge 1$.  
For a general signings $\epsilon,\epsilon' \in \{0, \pi\}^{\beta}$,
set 
\[Q_{\epsilon,\epsilon'}(h) ={\rm discr}\left((\epsilon*h) \oplus (\epsilon'*h)\right).\]  
The complement of $\mathcal B$ is contained in the algebraic set
\[ Z:= \bigcup_{\begin{smallmatrix}\epsilon, \epsilon' \in \{0,\pi\}^{\beta}\\\epsilon\ne \epsilon' \end{smallmatrix} } 
Q_{\epsilon,\epsilon'}^{-1}(0) = 
\left( \prod_{\begin{smallmatrix}\epsilon, \epsilon' \in \{0,\pi\}^{\beta}\\\epsilon\ne \epsilon' \end{smallmatrix}}
Q_{\epsilon,\epsilon'}\right)^{-1}(0).\]
The set $Z$ is a finite union of sets of the form $Q_{\epsilon}^{-1}(0)$.  
To see that each of these sets has codimension $\ge 1$ suppose otherwise.  Then there exists a signing $\epsilon$ so that
that $Q_{\epsilon}(h) = 0$ for all $h \in \mathcal S(G)$.

Choose a spanning tree in $G$.  Label the edges $e_1,e_2,\cdots, e_{\beta}$ in the complement and express $\epsilon = \sum \epsilon_ie_i$
as in \S \ref{spanning_tree} and \S \ref{subsec-combinatorics}.  Arrange the labeling so that $\epsilon_1 \ne 0$.
The graph $G'$ obtained from $G$ by removing
the edges $e_2, e_3, \cdots, e_{\beta}$ has $\beta(G') = 1$.  The signing $\epsilon$ of $G$ becomes a signing $\eta = \epsilon_1$
on $G'$, that is, a change of sign on the remaining edge $e_1$.  Moreover, any $h' \in \mathcal S(G')$ can be obtained as a limit
of $h \in \mathcal S(G)$ by allowing $h_{rs} \to 0$ where $(rs)$ varies over  the edges 
$e_2, e_3, \cdots, e_{\beta}$.  Since $Q_{\epsilon}(h)$ is a continuous function of $h$,
it vanishes on this limiting value, $h'$.  This proves that $Q_{\eta}(h') = 0$ for all $h' \in \mathcal S(G')$ which contradicts
the conclusion from the first paragraph. 
\end{proof}

\section{Proofs of Propositions \ref{prop-part1} and \ref{prop-part2}}\label{sec-proofs}

\subsection{Proof of Proposition \ref{prop-part1}}
Recall from \S \ref{subsec-outline} that $G$ is a simple connected graph with disjoint cycles,
$h \in \mathcal S(G)$ is generic in the sense of [GSC], and $\Lambda_k: \TT^{\beta} \to \RR$
is $\Lambda_k(\alpha) = \lambda_k(\alpha*h)$. 

It was shown in \cite{NodalCount, CdV, AG} that
each $\epsilon \in \{0,\pi\}^{\beta}\subset \TT^{\beta}$ is a nondegenerate critical point
of $\Lambda_k$ and its Morse index equals the nodal surplus.  In Corollary \ref{cor-diagonal}
it is shown that the Hessian of $\Lambda_k$ is diagonal with respect to the decomposition
 (\ref{eqn-torus-decomposition}).
Therefore the Morse index at $\epsilon\in \{0,\pi\}{^\beta}$ 
is the number of segments in the 1-skeleton of the
Boolean lattice that start at $\alpha$ and descend.  By the montonicity 
Lemma \ref{lem-monotonicity}, this is the same as the
number of segments whose endpoints have a lower eigenvalue, which is $|J_-(\alpha)|$.  \qed

\subsection{}
  A main tool that we will use in proving Proposition \ref{prop-part2}
is the local-global theorem of 
\cite{LocalTest}, which can be stated in a simplified manner as follows:
\begin{thm}\cite[Theorem 3.10]{LocalTest}\label{local-global}
Suppose $G$ is a simple, connected graph and $h\in\mathcal{S}(G)$ has a simple eigenvalue $\lambda_{k}(h)$ with a nowhere-vanishing eigenvector. Let $J\subset[\beta]$, let $\TT_{J}\subset\TT^{\beta}$ be the subtorus spanned by $\{e_{j}\}_{j\in J}$, and consider the restriction of $\Lambda_k$ to the subtorus $\TT_J$ (with $\Lambda_k(\alpha)
= \lambda_k(\alpha*h)$ as before).   Then, $\alpha=0$ is a local minimum (resp. maximum) of $\Lambda_{k}$ on $\TT_J$ if and only if it is a global minimum (resp. maximum) on $\TT_J$.
\end{thm}
The statements in \cite{LocalTest} involve a different but equivalent graph model, and apply in a situation of greater generality, where the eigenvector is permitted to vanish at various vertices. We therefore provide the proof for Theorem \ref{local-global}, adapted to our situation, in the Appendix. Theorem \ref{local-global} together with
 the monotonicity lemma gives the following:

\begin{cor}  \label{lem-Tplus}  Fix $h \in \mathcal S(G)$.
Fix $\alphab \in \sym$ and write $h_{\alphab} = \alphab*h$.
Let $\TT_-(\alphab)$ denote the sub-torus of $\TT^{\beta}$ that is spanned by those basis elements $e_j$ for
$j\in J_-(\alphab)$ and similarly for $\TT_+(\alphab)$.  Then 
\begin{align*}    
\lambda_k(\alpha*h_{\alphab}) &\le \lambda_k(h_{\alphab}) \text{ for any } \alpha \in \TT_-(\alphab)\\
\lambda_k(\alpha*h_{\alphab}) & \ge \lambda_k(h_{\alphab}) \text{ for any } \alpha \in \TT_+(\alphab).
\qed
\end{align*}  
\end{cor}

\quash{%%%% begin quash}
\begin{proof}
It suffices to prove this for $\alphab = 0$ since $h_{\alphab}$ is a signing of $h$.
Similarly take
$J_{\pm}= J_{\pm}(\epsilon)$.  Then $\alpha = 0$ is a critical point of $\Lambda_k(\alpha) = \lambda_k(\alpha*h)$.
Since the Hessian of $\Lambda_k(\alpha) = \lambda_k(\alpha*h)$  is diagonal (Proposition \ref{prop-part2}) with 
respect to the decomposition
(\ref{eqn-torus-decomposition}), the Morse index of $\Lambda_k$ is the number of directions on the
1-skeleton in which $\Lambda_k$ is locally decreasing.  By the monotonicity lemma, this equals the number
of neighboring vertices $\epsilon*h$ for which $\lambda_k(\epsilon*h) < \lambda_k(h)$.  Therefore these directions
span the sub-torus $\TT_-=\TT_{J_{-}}$.  Theorem \ref{local-global} guarantees that $\lambda_k(\alpha*h) \le \lambda_k(h)$ 
for all $\alpha \in \TT_{J_-}$.  The case of $\TT_+$ is similar.\end{proof}
} %%%% end quash

\subsection{Proof of Proposition \ref{prop-part2}}
Suppose $G$ is simple, connected, and has disjoint cycles, and suppose that $h$
is generic in the sense of [GSC].  
%Suppose also that $\lambda_k$ takes distinct values on distinct gauge-equivalence classes of signings of $h$.
Let $\alphab, \alphab' \in \sym$.  We need to show that
\begin{equation}\label{eqn-Jplus} J_+(\alphab) = J_+(\alphab') \iff \alphab = \alphab'.  \end{equation} 
The definition of $J_{\pm}(\epsilon)$ implicitly requires a choice of $k \in [n]$ and $h \in \mathcal S(G)$ 
so to be explicit we temporarily denote it $J_{\pm}(\epsilon, k, h)$.  For fixed $\epsilon, k$ this 
set is constant (in $h$) on connected components of the open set 
$\mathcal O$ of $h \in \mathcal S(G)$ which satisfy
condition [GSC], because the eigenvalues $\lambda_k(h), \lambda_k(\epsilon*h)$ vary continuously with $h$. As a result, it is enough to prove the statement for $h\in \mathcal{B}\cap\mathcal{O}$ as this is set is dense in $\mathcal{O}$ by Lemma \ref{lem-dag}. Recall that $h\in \mathcal{B}\cap\mathcal{O}$ if and only if it satisfies [GSC] and the 
condition (\maltese) which we repeat here:
\begin{enumerate}
    \item[(\maltese)]  For each $k \in [n]$ the eigenvalue $\lambda_k$ takes distinct values on distinct 
    gauge-equivalence classes of signings of $h$.
\end{enumerate}
Thus we may assume that $h$ satisfies [GSC] and (\maltese).  
Given $\alphab, \alphab' \in \sym$ suppose $J_+(\alphab) = J_+(\alphab')$. Assume for the sake of contradiction that  $\alphab\ne\alphab'$ so $\Lambda_{k}(\alphab)\ne\Lambda_{k}(\alphab')$. Assume that $\Lambda_{k}(\alphab)<\Lambda_{k}(\alphab')$ and let us show that there is $\alphab''$ such that $\Lambda_{k}(\alphab)>\Lambda_{k}(\alphab'')>\Lambda_{k}(\alphab')$ which provides the needed contradiction. Since $J_+(\alphab) = J_+(\alphab')$ then the intersection $\TT_{+}(\alphab')\cap\TT_{-}(\alphab)$ contains a signing, call it  $\alphab''$. Then, 
 Corollary  \ref{lem-Tplus} implies $\Lambda_k(\alphab'') > \Lambda_k(\alphab')$ because $\alphab''\in\TT_+(\alphab')$. However, $\Lambda_k(\alphab'') < \Lambda_k(\alphab)$ because $\alphab''\in\TT_{-}(\alphab)$.\qed

\appendix
\section{Proof of Theorem \ref{local-global}}
\subsection{}

We follow the proof in \cite{LocalTest} but reorder the steps. 
Theorem \ref{local-global} begins with a real symmetric matrix $h\in \mathcal S(G)$.
Recall that the choice of edge $(r_j,s_j)\in \gamma_j$ determines a basis $e_1,e_2,\cdots,e_{\beta}$ of $\TT^{\beta} = \RR^{\beta}/(2\pi i \ZZ)^{\beta}$.
The subset $J \subset [\beta]$ determines the subtorus $\TT_J 
\subset \TT^{\beta}$ which is spanned by the coordinates $e_j$ for $j \in J$. 
We therefore have an analytic
family of magnetic perturbations, $h_{\alpha} = \alpha*h$ for $\alpha \in \TT_{J}$, 
and an eigenvalue function
$\Lambda_{k}:\TT_{J}\to\RR$ defined by $\Lambda_{k}(\alpha):=\lambda_{k}(h_{\alpha})$. Since $\lambda = \lambda_k(h)$ is a simple eigenvalue, 
the function $\Lambda_k$ is analytic near $\alpha = 0$ and is piecewise analytic on all of $\TT_J$.  We may choose the corresponding eigenvector $\phi$ of $h$ to be real. By assumption, it is nowhere vanishing.

The point $\alpha = 0$ is a critical point of $\Lambda_k$.  Assume it is a local minimum. Theorem \ref{local-global} states that it is 
also a global minimum.  (The case of a maximum can be proven analogously.) So we need to show 
\begin{equation}\label{eq: global min}
    \lambda\le \lambda_{k}(h_{\alpha})\ \text{ for all}\ \alpha\in\TT_{J}.
\end{equation}

\subsection{}
The proof in \cite{LocalTest} involves several auxiliary matrices.
 Holding $\phi$ constant, the function $\langle \phi, h_{\alpha}\phi\rangle:\TT_{J} \to \RR$ has a
critical point at $\alpha = 0$ (cf. equation (\ref{eqn-partial-derivative})) and we set
\[ \Omega = \left.\frac{1}{2}\Hess\left( \langle \phi, h_{\alpha}\phi \rangle \right)\right|_{\alpha=0}.\]
The matrix $\Omega$ is a real diagonal $|J|\times|J|$ matrix. It is diagonal since each entry of $h_{\alpha}$ depends on at most one $\alpha_{j}$ coordinate, so $\frac{\partial^2 h_{\alpha}}{\partial \alpha_{i}\partial \alpha_{j}}=0$ for $i\ne j$. It is real and invertible since its diagonal entries are 
\begin{equation}\label{eqn-Omegajj}
\Omega_{jj} = -h_{r_js_j}\phi(r_j)\phi(s_j)\ne 0\end{equation}
as calculated in (\ref{eqn-ddoth}). (Recall that both $h$ and $\phi$ are real, and $\phi$ is nowhere-vanishing.)

 For each $j\in J$ let $R_{j}(t)$ be the hermitian $n \times n$ matrix supported on the block 
 \[ \left[r_{j}r_{j},r_{j}s_{j};s_{j}r_{j},s_{j}s_{j}\right]\]
 on which it is given by
    \begin{equation}\label{eq 2 by 2}
 R_j(t)\quad  =\quad    h_{r_js_j}   \begin{pmatrix}
 -\frac{\phi(s_{j})}{\phi(r_{j})} & e^{it}\\
 e^{-it} &  -\frac{\phi(r_{j})}{\phi(s_{j})}
      \end{pmatrix}. 
    \end{equation}
To ease notation let us assume that $J = \{1, 2, \cdots, |J|\}$. Writing
$\alpha = (\alpha_1, \alpha_2, \cdots, \alpha_{J}) \in \TT_J$, the sum
\[ \sum_{j\in J}R_j(\alpha_j)\]
is a family of Hermitian $n \times n$ matrices depending on $\alpha \in \TT_J$.

    Define the real symmetric $n\times n$ (constant) matrix $S$ by
    \begin{equation}\label{S}
        S=h - \sum_{j\in J}
        R_j(0).
    \end{equation}
%Let $S=h-\sum_{j=1}^{\beta}R_{j}(0)$. 
This collection of matrices satisfies the following properties:
\begin{enumerate}[label=(\alph*)]
     \item For any $\alpha =(\alpha_1,\alpha_2,\cdots,\alpha_{|J|}) \in \TT_{J}$, 
\[ h_{\alpha}= \alpha*h = S + \sum_{j\in J} 
        R_j(\alpha_j).\]
    \item For any $j\ne j'$, $R_{j}(t)$ and $R_{j'}(t')$ commute for all $t,t'$.
    \item $R_{j}(0)\phi=0$ for every $j\in J$, and hence $S\phi=\lambda\phi$.
      \item $\det(R_j(t)) = 0$ so $R_j(t)$ has rank one
\item The semi-definite sign of $R_j(t)$ is independent of $t$ since ${\rm trace}(R_j(t)) = 2\Omega_{jj}$.
  %  \item $\alpha'*h=S+\sum_{j=1}^{\beta}R_{j}(a_{j})$ for any $\alpha'=\sum_{j=1}^{\beta}a_{j}e_{j}$.
\end{enumerate}
(see equation (\ref{eqn-Omegajj})).  Let 
\[m=\left\vert\{j\in J\ :\  -h_{r_is_i}\phi(r_i)\phi(s_i)<0\}\right\vert= \mathrm{ind}(\Omega)\]
be the number of negative semi-definite $R_{j}$'s. Then the sum of these commuting rank-one matrices has $m$ negative eigenvalues and $n-|J|>0$ zero eigenvalues (recalling that $|J|\le \beta<n$ by the assumption of disjoint cycles), so    
\[\lambda_{m+1}\left(\sum_{j\in J} 
        R_j(\alpha_{j})\right)=0\ \text{ for all }\ \alpha \in \TT_{J} .\]
        The Weyl inequalities for $ h_{\alpha}= S + \sum_{i=1}^{|J|} 
        R_j(\alpha_j)$ may be expressed as follows,
\begin{diagram}[size=1.3em]
\lambda_p(S) + \lambda_q(\Sigma R_j)\ & \le\ & \lambda_k(S+\Sigma R_j)\ & \le\ & \lambda_s(S) + \lambda_r(\Sigma R_j) \\
(p+q\le k+1) && && (k+n \le r+s )
\end{diagram}
Only the first inequality is required for the case of a local minimum.  Taking $q = m+1$ gives
        \begin{equation}\label{eq: Weyl enq}
           \lambda_{k-m}(S)
% = \lambda_{k-m}(S)+\lambda_{m+1}\left(\sum_{i=1}^{|J|}         R_j(\alpha_{j})\right)
        \le \lambda_{k}(h_{\alpha})\ \text{ for all }\ \alpha \in \TT_{J}.
        \end{equation}
By \eqref{eq: global min} the proof of Theorem \ref{local-global} now comes down to the following statement:
\begin{lem}\label{S lemma}
    If $\alpha = 0$ is a local minimum of $\Lambda_k(\alpha)$ then $\lambda_{k-m}(S)=\lambda_{k}(h)=\lambda$. 
\end{lem}
The proof involves the next few paragraphs.

\subsection{ }

Holding $\phi$ constant gives a mapping $ih_{\alpha}\phi:\TT_{J} \to \CC^n$. Define $B$ to be its 
derivative $B = \left. iD(h_{\alpha}\phi)\right|_{\alpha=0}$.  It is a real $n \times |J|$ matrix with   
\begin{equation*}\label{eqn-Bvj}
B_{vj} = \left.\frac{\partial}{\partial \alpha_j} \left(h_{\alpha}\phi  \right)_v\right|_{\alpha=0}
= \begin{cases}
    -h_{r_js_j}\phi(r_j) &\text{ if } v=s_j\\
    h_{s_jr_j} \phi({s_j}) &\text{ if } v=r_j\\
    0 & \text{ otherwise}
\end{cases}
\end{equation*}
A direct but messy calculation involving double subscripts as in \cite[Lemma 2.7]{LocalTest} gives
\begin{equation}\label{eq: S as shur}
    \sum_{j\in J}R_{j}(0)=B\Omega^{-1}B^{T},\ \text{ and therefore }\ S=h-B\Omega^{-1}B^{T}.
\end{equation}

\subsection{}A primary insight in \cite{LocalTest} is the identification of the generalized Schur complements
in the real symmetric $(n+|J|)\times(n+|J|)$ matrix
\begin{equation}\label{eqn-M}
M = \left(\begin{matrix}
 h-\lambda& B \\B^T & \Omega 
\end{matrix}\right).
\end{equation}
These complements are defined to be
\begin{align*}
    M/(h-\lambda) = \Omega - B^T(h-\lambda)^+B \\
M/\Omega = (h-\lambda) - B \Omega^{-1}B^T=S-\lambda.\end{align*}
where ``$+$" denotes the Moore-Penrose pseudo-inverse.\footnote{The Moore-Penrose pseudo-inverse of a real symmetric matrix $A$ is zero on $(\rm{Im}(A))^{\perp}$ and is the inverse of the isomorphism $(\ker(A))^{\perp} \to {\rm{Im}}(A)$.}
 
\begin{prop}\cite[Lemma 2.3]{LocalTest}\label{prop-Schur} The Schur complement to $h-\lambda$ may be identified,
\[ M/(h-\lambda)= \frac{1}{2}\Hess(\Lambda_{k}(0)).\]
\end{prop}
The proof in \cite{LocalTest} requires \cite{Kato} (Remark II.2.2 p.~81) but it is actually elementary and we provide it here for completeness.   
The Lemma is equivalent to the statement that
\[ \frac{1}{2}\langle \eta, \Hess(\Lambda_{k}(0))\eta\rangle = \langle \eta,\Omega\eta\rangle-\langle B\eta,(h-\lambda)^+B\eta\rangle\ \text{ for all } \eta \in T_{0}\TT^{J}=\RR^{J}.\]

To calculate $\langle \eta, \Hess(\Lambda_{k}(0) \eta \rangle$, choose an analytic one parameter
family $\alpha_t$ with $\eta = \dot\alpha(0)$ and write $h_t = \alpha_t*h$ with simple eigenvalue $\Lambda_{k}(\alpha_t)$ and normalized eigenvector $\phi_t$ (so that $\phi = \phi_0$). 
From (\ref{eqn-first-derivative}) and (\ref{ddot-mu}) the second derivative is
\[\langle \eta, \Hess(\Lambda_{k}(0) \eta \rangle=\left. \frac{d^2}{dt^2}\Lambda_{k}(\alpha_t)\right|_{t=0} =\left. \frac{d^2}{dt^2}\left(\langle \phi_t, h_{t}\phi_t \rangle\right)\right|_{t=0}= \left.\langle \phi, \ddot h \phi \rangle + 2 \Re[\langle \phi, \dot h \dot \phi \rangle]\right|_{t=0}
,\]
where $\dot h = \frac{d}{dt}h_t|_{t=0},\ \dot \phi = \frac{d}{dt} \phi_t|_{t=0} $, and $\ddot h = \ddot h_t|_{t=0}$ (This is the formula from \cite{Kato} that is referenced in \cite{LocalTest}.)  The first term agrees with the first term in $2\langle\eta, ( M/h-\lambda) \eta \rangle$: 
\[ \frac{1}{2}\langle \phi, \ddot h \phi \rangle =
\frac{1}{2}\frac{d^2}{dt^2}\left. \langle \phi, h_t \phi \rangle\right|_{t=0} =
\langle \eta, \Omega \eta \rangle.
\]
The $t$-derivative of $ih_t\phi$ (keeping $\phi$ fixed) is
\[B\eta= iD_{\eta}(h_{\alpha}\phi) = i\dot h\phi\]
 So we need to compare
\[
-\langle \eta, B^*(h - \lambda)^+ B \eta \rangle = 
-\langle \dot h \phi, (h - \lambda)^+ \dot h \phi \rangle
\]
with $\langle \phi, \dot h \dot \phi \rangle =\langle \dot h \phi, \dot \phi \rangle$.  
From (\ref{eqn-first-derivative}),
\[ \dot \phi + (h-\lambda)^+\dot h \phi = c \phi\]
for some constant $c$, because $\phi$ spans the (one dimensional) kernel of $h-\lambda$.  Taking the inner product
with $\dot h \phi$ and using (\ref{eqn-partial-derivative}) with $\dot \lambda = 0$ gives
\[ \langle \dot h \phi, \dot \phi \rangle + \langle \dot h \phi, (h-\lambda)^+ \dot h \phi
\rangle = 0\]
as claimed. \qed
\subsection{Proof of Lemma \ref{S lemma}}
The Hainsworth theorem for the matrix $M$ in (\ref{eqn-M}) gives,
\[\mathrm{ind}(M)=\mathrm{ind}(M/(h-\lambda))+\mathrm{ind}(h-\lambda)=\mathrm{ind}(M/\Omega)+\mathrm{ind}(\Omega),\]
which yields
\[\mathrm{ind}(S-\lambda)=\mathrm{ind}(M/(h-\lambda))+k-1-m.\]
Since $\alpha = 0$ is a local minimum of $\Lambda_k$, Proposition \ref{prop-Schur}
gives $\mathrm{ind}\left(M/(h-\lambda)\right)= 0$.  Property (c) of the matrices
$R_j(t)$ implies $\lambda$ is an eigenvalue of $S$.  Therefore
$\lambda_{k-m}(S-\lambda) = 0$.\qed

\section{The BZ condition} \label{sec-BZ} 
The argument in Lemma \ref{lem-Jnot0}  concerning eigenvalues with nontrivial multiplicity is essentially
the same as that of Theorem 1.5 in \cite{GenericFamilies}, which we state here for completeness because
it is an important observation about singular critical points that may appear.
We are interested in the Morse theory of the composition $\Lambda_k:\TT^{\beta} \to \RR$,
\[ \begin{CD}
\TT^{\beta} @>>> \mathcal H(G) @>>{\lambda_k}> \RR. \end{CD}\]

Fix $\alpha \in \TT^{\beta}$ and suppose that $\lambda_k(\alpha*h)$ is an eigenvalue of multiplicity $m\le \beta$.  
Let $V$ denote the $m$-dimensional eigenspace.  Consider the set of all Hermitian forms on $V$ that are given by
\begin{equation}\label{eqn-TangentForm}
\langle \phi, \frac{d}{dt}[(\alpha+tv)*h] \nu \rangle \text{ for } \phi,\nu \in V\end{equation} 
as $v$ varies within $T_{\alpha}\TT^{\beta}$.  According to Theorem 1.5
in \cite{GenericFamilies}, {\em if there exists $v \in T_{\alpha}\TT^{\beta}$ such that
the form (\ref{eqn-TangentForm}) is positive definite} (which we refer to as the BZ condition), then the point
$\alpha \in \TT^{\beta}$ is topologically regular, meaning that for sufficiently small $\delta >0$ the set $\TT^{\beta}_{\le \lambda - \delta}$ is a strong deformation retract of $\TT^{\beta}_{\le \lambda + \delta}$.
(Here, $\TT^{\beta}_{\le t} = \left\{ \alpha' \in \TT^{\beta}:\ \Lambda_k(\alpha')\le t\right\}$ and
$\lambda = \Lambda_k(\alpha)$.)

\end{document}